\documentclass[11pt,a4paper]{article}
\usepackage{amsmath,amsthm,amsfonts,graphicx,hyperref,color,caption,subcaption,float,fullpage,algpseudocode,soul}
\usepackage[margin=1in]{geometry}

\usepackage{microtype,amssymb}
\usepackage{mathtools}
\usepackage{verbatim}
\usepackage{authblk}

\DeclarePairedDelimiter\floor{\lfloor}{\rfloor}
\usepackage{framed}
\newtheorem{thm}{Theorem}
\newtheorem{lem}[thm]{Lemma}
\newtheorem{fct}[thm]{Fact}
\newtheorem{cor}[thm]{Corollary}
\newtheorem{dfn}[thm]{Definition}

\newtheorem{clm}[thm]{Claim}

\def\RR{{\mathbb R}}
\def\EE{{\mathbb E}}
\def\SS{{\mathbb S}}
\def\Pr{{\mathrm Pr}}
\DeclarePairedDelimiter\abs{\lvert}{\rvert}
\DeclarePairedDelimiter\norm{\lVert}{\rVert}

\begin{document}
\title{High-dimensional approximate $r$-nets}
\author{Georgia Avarikioti\thanks{School of Electrical and Computer Engineering, National Technical University of Athens, Athens, Greece.\if 0 E-mail:  zetavar@corelab.ntua.gr \fi} \qquad Ioannis Z.~Emiris\thanks{Department of Informatics \& Telecommunications, University of Athens, Athens, Greece. \if 0 E-mail: emiris@di.uoa.gr \fi} \qquad Loukas Kavouras\thanks{School of Electrical and Computer Engineering, National Technical University of Athens, Athens, Greece. \if 0 E-mail: lukassemfe@hotmail.com \fi } \qquad Ioannis Psarros\thanks{Department of Informatics \& Telecommunications, University of Athens, Athens, Greece. \if 0E-mail: ipsarros@di.uoa.gr \fi}}

\maketitle

\begin{abstract}
The construction of $r$-nets offers a powerful tool in computational and metric geometry.
We focus on high-dimensional spaces and
present a new randomized algorithm which efficiently computes approximate $r$-nets with respect to Euclidean distance.
For any fixed $\epsilon>0$, 
the approximation factor is $1+\epsilon$ and the complexity is polynomial in the dimension and subquadratic in the number of points.  The algorithm succeeds with high probability.  More specifically, 
the best previously known LSH-based construction of Eppstein et al.\ \cite{EHS15} is improved in terms of complexity by reducing the dependence on $\epsilon$, 
provided that $\epsilon$ is sufficiently small.
Our method does not require LSH but, instead, follows Valiant's \cite{Val15} approach in designing a sequence of reductions of our problem to other problems in different spaces, under Euclidean distance or inner product, for which $r$-nets are computed efficiently and the error can be controlled.
Our result immediately implies efficient solutions to a number of geometric problems in high dimension, such as finding the $(1+\epsilon)$-approximate $k$th nearest neighbor distance in time subquadratic in the size of the input. 

\bigskip

\end{abstract}

\section{Introduction}

We study $r$-nets, a powerful tool in computational and metric geometry, with several applications in approximation algorithms.
An $r$-net for a metric space $(X,\norm{\cdot}), \, |X|=n$ and for numerical parameter $r$ is a subset $R\subseteq X$ 
such that the closed $r/2$-balls centered at the points of $R$ are disjoint, and the closed $r$-balls around the same points cover all of $X$. We define approximate $r$-nets analogously. Formally,

\begin{dfn}
Given a pointset $X \subseteq \RR^d$, a distance parameter $r \in \RR$ and an approximation parameter $\epsilon >0$, a $(1+\epsilon)r$-net of $X$ is a subset $R \subseteq X$ s.t. the following properties hold:
\begin{enumerate}
\item (packing) For every $p, q \in R$, $p \neq q $, we have that $\norm{p-q}_2 \geq r$.
\item (covering) For every $p \in X$, there exists a $q \in R$ s.t. $\norm{p-q}_2 \leq (1+\epsilon)r$.
\end{enumerate}
\end{dfn}

\paragraph{Previous Work.} 
Finding $r$-nets can be addressed naively by considering all points of $X$ unmarked and, while there remains an unmarked point $p$, the algorithm adds it to $R$ and marks all other points within distance $r$ from  
$p$. The performance of this algorithm can be improved by using grids and hashing \cite{HP04}. 
However, their complexity remains too large when dealing with big data in high dimension. The naive algorithm is quadratic in $n$ and the grid approach is in $O(d^{d/2} n)$, hence it is relevant only for constant dimension $d$~\cite{HR15}. 
In \cite{HPM05}, they show that an approximate net hierarchy for an 
arbitrary finite metric can be computed in $O(2^{ddim}n \log n)$, where $ddim$ is the doubling dimension.
This is satisfactory when doubling dimension is constant, but requires a vast amount of resources when it is high.

When the dimension is high, there is need for algorithms with time complexity polynomial in $d$ and subquadratic in $n$. One approach, which computes $(1+\epsilon)r$-nets in high dimension is that of \cite{EHS15}, which uses the Locality Sensitive Hashing (LSH) method of
\cite{AI08}. The resulting time complexity is
$\tilde{O}(dn^{2-\Theta(\epsilon)})$, where $\epsilon>0$ is quite small and $\tilde{O}$ hides polylogarithmic factors. 

In general, high dimensional analogues of  classical geometric problems 
have been mainly addressed by LSH. For instance, the approximate closest pair problem can be trivially solved by performing $n$ approximate nearest neighbor (ANN) queries. For sufficiently small $\epsilon$, this costs $\tilde{O}(dn^{2-\Theta(\epsilon)})$ time, due to the complexity factor of an LSH query. Several other problems have been reduced to ANN queries \cite{GIV01}. Recently, Valiant \cite{Val12}, \cite{Val15} presented an algorithm for the approximate closest pair problem 
in time $\tilde{O}(dn^{2-\Theta(\sqrt{\epsilon})})$. This is a different approach in the sense that while LSH exploits dimension reduction through random projections, the algorithm of \cite{Val15} is inspired by high dimensional phenomena. One main step of the algorithm is that of projecting the pointset up to a higher dimension.

\paragraph{Our Contribution.} 
We present a new randomized algorithm that computes approximate $r$-nets in time subquadratic in $n$ and polynomial in the dimension, and improves upon the complexity of the best known algorithm.
Our method does not employ LSH and, with probability $1-o(1)$, it returns 
$R\subset X$, which is a $(1+\epsilon)r$-net of $X$. 

We reduce the problem of an approximate $r$-net for arbitrary vectors (points) under Euclidean distance to the same problem for vectors on the unit sphere. Then, depending on the magnitude of distance $r$, an algorithm handling ``small" distances or an algorithm handling ``large" distances is called. 
These algorithms reduce the Euclidean problem of $r$-nets on unit vectors to finding an $r$-net for unit vectors under inner product (Section~\ref{SGeneral}). This step requires that the multiplicative $1+\epsilon$ approximation of the distance corresponds to an additive $c\epsilon$ approximation of the inner product, for suitable constant $c>0$.

Next, we convert the vectors having unit norm into vectors with entries $\{-1, +1\}$ (Section \ref{SInner}).
This transformation is necessary in order to apply the Chebyshev embedding of \cite{Val15}, an embedding that damps the magnitude of the inner product of ``far" vectors, while preserving the magnitude of the inner product of ``close"
vectors. 
For the final step of the algorithm, we first 
apply a procedure that allows us to efficiently 
compute $(1+\epsilon)$-nets in the case where the number of ``small" distances is large. Then, we apply a modified version of the {\tt Vector Aggregation} algorithm of \cite{Val15}, that exploits fast 
matrix multiplication, so as to achieve the desired running time.

In short, we extend Valiant's framework \cite{Val15} and we compute $r$-nets in time $\tilde{O}(dn^{2-\Theta(\sqrt{\epsilon})})$, thus improving on the exponent of the LSH-based construction \cite{EHS15}, when $\epsilon$ is small enough. This improvement by $\sqrt{\epsilon}$ in the exponent is the same as the complexity improvement obtained in \cite{Val15} over the LSH-based algorithm for the approximate closest pair problem. 

Our study is motivated by the observation that computing efficiently an $r$-net leads to efficient solutions for several geometric problems, specifically in approximation algorithms. In particular, our extension of $r$-nets in high dimensional Euclidean space can be plugged in the framework of~\cite{HR15}. The new framework has many applications, notably the $k$th nearest neighbor distance problem, which we solve in $\tilde{O}(dn^{2-\Theta(\sqrt{\epsilon})})$.


\paragraph{Paper Organization.}
Section \ref{SInner} presents an algorithm for 
computing an approximate net with respect to the inner product for a set of unit vectors.
Section~\ref{SGeneral} translates the problem of finding $r$-nets under Euclidean distance to the same problem under inner product. In Section \ref{Sapps}, we discuss applications of our construction and possible future work. Omitted proofs are included in the Appendices.

We use $\norm{\cdot}$ to denote the Euclidean norm $\norm{\cdot}_2$ throughout the paper. 

\section{Points on a sphere under inner product}\label{SInner}
In this section, we design an algorithm for constructing an approximate $\rho$-net of vectors on the sphere under inner product.
To that end, we reduce the problem to constructing an approximate net under absolute inner product for vectors that lie on the vertices of a unit hypercube.

Since our ultimate goal is a solution to computing $r$-nets with respect to Euclidean distance, we allow additive error in the approximation, which under certain assumptions, translates to multiplicative error in Euclidean distance. 
In the following, we define rigorously the notion of approximate $\rho$-nets under inner product.

\begin{dfn}
\label{DfnNetInn}
For any $X\subset \SS^{d-1}$, an approximate $\rho$-net for $(X,\langle \cdot,\cdot\rangle)$ , 
 with additive approximation parameter $\epsilon>0$, is a subset $C\subseteq X$ which satisfies the following properties:
 \begin{itemize}
 \item for any two $p \neq q \in C$, $\langle p, q \rangle < \rho$, and
  \item for any $x\in X$, there exists $p \in C$ s.t. $\langle x, p \rangle \geq \rho-\epsilon$.
 \end{itemize}

\end{dfn}
One relevant notion is that of $\epsilon$-kernels \cite{AHV05}. In $\epsilon$-kernels, one is interested in finding a subset of the input pointset, which approximates its directional width. Such constructions have been extensively studied when the dimension is low, due to their relatively small size.  

\subsection{Crude approximate nets}

In this subsection we develop our basic tool, which is based on the Vector Aggregation Algorithm by \cite{Val15}. This tool aims to compute  approximate $\rho$-nets with multiplicative error, as opposed to what we have set as our final goal for this section, namely to bound additive error. 
Moreover, in the context of this subsection, two vectors are close to each other when the magnitude 
of their inner product is large, and two vectors are 
far from each other when the magnitude of their inner product is small.
Let $|\langle \cdot,\cdot \rangle|$ denote the magnitude of the inner product of two vectors.

\begin{dfn} \label{DfnMagnInnNet}
 For any $X=[x_1,\ldots, x_n], X' =[x_1',\ldots,x_n'] \subset \RR^{d \times n}$, a crude approximate $\rho$-net for $(X,X',|\langle \cdot,\cdot \rangle|)$, 
 with multiplicative approximation factor $c>1$, is a subset $C\subseteq [n]$ which satisfies the following properties:
 \begin{itemize}
    \item for any two $i \neq j \in C$, $|\langle x_i, x_j' \rangle| < c \rho$, and
\item for any $i\in [n]$, there exists $j \in C$ s.t.\ $|\langle x_i, x_j' \rangle| \geq \rho$.
 \end{itemize}

\end{dfn}

{\tt Vector Aggregation} follows the 
exposition of \cite{Val15}. 
The main difference is that, instead of the ``compressed'' matrix $Z^T Z$, we use the form $X^T Z$, where $Z$ derives from vector aggregation.
Both forms encode the information in the Gram matrix $X^T X$.
The matrix $X^TZ$ is better suited for our purposes, since each row corresponds to an input vector instead of an aggregated subset; this extra information may be useful in further problems.

\begin{framed}

{\tt Vector Aggregation}\\

Input: $X =[x_1,\ldots,x_n] \in \RR^{d \times n}$, 
$X' =[x_1',\ldots,x_n'] \in \RR^{d \times n}$, 
$\alpha \in (0,1)$, $\tau>0$. 

Output: $n\times n^{1-\alpha}$ matrix $W$ and random partition $S_1 , \ldots , S_{n^{1-\alpha}}$ of $\{x_1,\ldots,x_n\}$.

  \begin{itemize}
   \item Randomly partition $[n]$ into $n^{1-\alpha}$ disjoint subsets, each of size $n^{\alpha}$ , denoting the
sets $S_1 , \ldots , S_{n^{1-\alpha}}$.
   \item For each $i = 1, 2, \ldots , 78 \log n$:
   \begin{itemize}
    \item Select $n$ coefficients $q_1 ,\ldots , q_n \in \{-1, +1\}$ at random.
    \item Form the $d\times n^{1-\alpha}$ matrix $Z^i$ with entries $z_{j,k}^i=\sum_{l\in S_k} q_l \cdot x_{j,l}'$ 
    \item $W^i=X^T Z^i$
   \end{itemize}

   \item Define the $n \times n^{1-\alpha}$ matrix $W$ with $w_{i,j} =quartile(|w_{i,j}^1|,\ldots|w_{i,j}^{78 \log n}|)$. 
   \item Output $W$ and $S_1 , \ldots , S_{n^{1-\alpha}}$.
  \end{itemize}
 
\end{framed}

\begin{thm}\label{ThmVecAgg}
Let $X \in \RR^{d \times n}$, 
$X'  \in \RR^{d \times n}$, 
$\alpha \in (0,1)$, $\tau>0$ the input of {\tt Vector Aggregation}. Then, the algorithm returns a matrix $W$ of size $n\times n^{1-\alpha}$ and a random partition $S_1 , \ldots , S_{n^{1-\alpha}}$, which with probability $1-O(1/n^3)$ satisfies the following:
\begin{itemize}
\item For all $j\in [n]$ and $k\in [n^{1-\alpha}]$, if $\forall u \in S_k$, $ |\langle x_j, u \rangle|\leq \tau$ then $|w_{j,k}| < 3 \cdot n^{\alpha} \tau$.
\item For all $j\in [n]$ and $k\in [n^{1-\alpha}]$ 
if $\exists u\in S_k$, $|\langle x_j, u \rangle|\geq 3n^{\alpha}\tau$ then $|w_{j,k}| \geq 3 \cdot n^{\alpha} \tau$.
\end{itemize}
Moreover, the algorithm runs in time 
$\tilde{O}(dn+n^{2-\alpha}+MatrixMul( n \times d,d \times n^{1-\alpha}))$.
\end{thm}

For the case of pointsets with many ``small" distances, 
we rely crucially on the fact that the expected 
number of near neighbors for a randomly chosen point 
is large. 
So, if we iteratively choose random points and 
 delete these and their neighbors, we will end up with 
 a pointset which satisfies the property of having sufficiently few ``small" distances. Then, we apply {\tt Vector Aggregation}.

\begin{framed}
{\tt Crude ApprxNet}\\

Input: $X =[x_1,\ldots,x_n] \in \RR^{d \times n}$, 
$X' =[x_1',\ldots,x_n'] \in \RR^{d \times n}$, 
$\alpha \in (0,1)$, $\tau>0$. 

Output: $C'\subseteq [n]$, {$F' \subseteq [n]$}.

\begin{itemize}
\item $C\gets \emptyset$, $F_1 \gets \emptyset, F_2 \gets \{x_1,\ldots,x_n\}$
\item Repeat $n^{0.5}$ times:
 \begin{itemize}
 \item Choose a column $x_i$ uniformly at random.
 \item $C \gets C \cup \{x_i\}$.
  \item Delete column $i$ from matrix $X$ and column $i$ from matrix $X'$.
 \item Delete each column $k$ from matrix $X$, $X'$ s.t. $|\langle x_i, x_k'
 \rangle| \geq \tau$.
  \item If there is no column $k$ from matrix $X$ s.t. $|\langle x_i, x_k'\rangle| \geq \tau$, then $F_1 \gets F_1 \cup \{x_i\}$
  \end{itemize}

\item Run {\tt Vector Aggregation} 
with input $X$, $X'$, $\alpha$, $\tau$ and output 
$W$, $S_1,\ldots,S_{n^{1-\alpha}}$.
\item For each of the remaining columns $i=1,\ldots$:
   \begin{itemize}
    \item For any $|w_{i,j}|\geq 3 n^{\alpha} \tau$:
     \begin{itemize}
     \item If more than $n^{1.7}$ times in here, output "ERROR".
      \item Compute inner products between $x_i$ and vectors in $S_j$. 
      For each vector $x_k' \in S_j$ s.t. $x_k' \neq x_i$ and $|\langle x_i,x_k'\rangle|\geq \tau$, 
      delete row $k$ {and $F_2 \gets F_2 \backslash \{ x_i\}$.}
     \end{itemize}
     \item $C \gets C \cup \{x_i\} $ 
   \end{itemize}
\item Output indices of $C$ {and $F \gets \{F_1 \cup F_2 \}$}.
\end{itemize}

\end{framed}

\begin{thm}\label{ThmCrudeNet}
  On input $X =[x_1,\ldots,x_n] \in \RR^{d \times n}$, 
$X' =[x_1',\ldots,x_n'] \in \RR^{d \times n}$, 
$\alpha \in (0,1)$, $\tau>0$, {\tt Crude ApprxNet}, computes a crude  $3n^{\alpha}$-approximate $\tau$-net for $X$, $X'$, following the notation of Definition \ref{DfnMagnInnNet}. 
The algorithm costs time:
$$
\tilde{O}(n^{2-\alpha}+ d \cdot n^{1.7+\alpha}+MatrixMul( n \times d,d \times n^{1-\alpha})),
$$
and succeeds with probability $1-O(1/n^{0.2})$. 
Additionally, it outputs a set $F\subseteq R$ with the following property:
$\{x_i \mid \forall x_j \neq x_i~ |\langle x_j,x_i \rangle | <  \tau \}\subseteq F \subseteq \{x_i \mid \forall x_j \neq x_i~ |\langle x_j,x_i \rangle | < n^a\tau \}$.
\end{thm}

\begin{proof}

We perform $n^{0.5}$ iterations and for each, we compare the inner products between the randomly chosen vector and all other vectors. Hence, the time needed is $O(dn^{1.5})$.

In the following, we denote by 
$X_i$ the number of vectors which have ``large" magnitude of the inner product with the randomly chosen point in the $i$th iteration. 
Towards proving correctness, 
suppose first that $\EE[X_i]>2n^{0.5}$ for all $i=1, \ldots n^{0.5}$. The expected number of vectors we delete in each 
iteration of the algorithm is more than $2n^{0.5}+1$. So, after $n^{0.5}$ iterations,
the expected total number of deleted vectors will be greater than $n$. This means
that if the hypothesis holds for all iterations we 
will end up with a proper net. 

Now suppose that there is an iteration $j$ where $\EE[X_j] \leq 2n^{0.5}$. After all iterations, the 
number of ``small" distances are at most $n^{1.5}$ on expectation. By Markov's inequality, 
when the 
{\tt Vector Aggregation} algorithm is called, the following is satisfied
with probability 
$1-n^{-0.2}$ : 
$$
|\{(i,k) \mid |\langle x_i, x_k'\rangle| \geq\tau, i\neq k\}| \leq n^{1.7} .
$$
By Theorem \ref{ThmVecAgg}
and the above discussion, the number of entries in the matrix $W$ that we need to visit is at most $n^{1.7}$. For each entry, we perform a brute force 
which costs $d n^\alpha$. 

Now notice that the first iteration stores centers $c$ and deletes all points $p$ for which $|\langle c,p\rangle| \geq \tau$. Hence,  any two centers $c,c'$ satisfy $ |\langle c,p\rangle| < \tau$. In the second iteration, over the columns of $W$, notice that by Theorem \ref{ThmVecAgg}, for any two centers $c,c'$ we have $|\langle c,c'\rangle| <3 n^{\alpha}\tau.$
\end{proof}


\subsection{Approximate inner product nets}

In this subsection, we show that the problem of computing $\rho$-nets for the inner product of unit 
vectors reduces to the less natural problem of Definition \ref{DfnMagnInnNet}, which refers to the magnitude of the inner product.

The first step consists of mapping the unit vectors to vectors in $\{-1,1\}^{d'}$. The mapping is essentially Charikar's LSH scheme \cite{Cha02}. 
Then, we apply the Chebyshev embedding of~\cite{Val15} in order to achieve gap amplification, and finally 
we call algorithm {\tt Crude ApprxNet}, which will now return 
a proper $\rho$-net with additive error.  
\begin{thm}[\cite{Val15}] \label{ThmUnif}
There exists an algorithm with the following properties. Let $d'=O(\frac{\log n}{\delta^2})$
and $Y \in \RR^{d'\times n}$ denote its output 
on input $X$, $\delta$, where $X$ is a
matrix whose columns have unit norm, with probability $1 - o(1/n^2 )$, for all pairs $i, j \in [n]$,
$
\Big|{\langle Y_i , Y_j \rangle}/{d'}-\Big(1-2 \cdot {\mathrm{cos}^{-1}(\langle X_i,X_j \rangle)}/{ \pi}\Big)\Big| \leq \delta
,$
where $X_i$, $Y_i$ denote the $i$th column of $X$ and $Y$ respectively. Additionally, the runtime of 
the algorithm is $O( \frac{d n \log n}{\delta^2})$.
\end{thm}

The following theorem provides a randomized embedding that damps the magnitude of the inner product of ``far" vectors, while preserving the magnitude of the inner product of ``close"
vectors. The statement is almost verbatim that of~\cite[Prop.6]{Val15} except that we additionally establish an asymptotically better probability of success. The proof is the same, but since we claim stronger guarantees on success probability, we include the complete proof in Appendix~\ref{AppThmCheb}.

\begin{thm}\label{ThmCheb}
Let $Y$, $Y'$ be the matrices output by algorithm ``Chebyshev Embedding" on
input $X, X' \in \{-1,1\}^{d\times n} , \tau^{+}\in [-1,1] , \tau^{-} \in [-1,1]$ with $\tau^{-}<\tau^{+}$ , integers $q, d'$. 
With probability
{$ 1 - o(1/n)$} over the randomness in the construction of
$Y, Y'$, for all $i, j \in [n]$, 
$\langle Y_i , Y_j' \rangle$ is within $\sqrt{d'} \log n$ from the value 
$
T_q \Big(\frac{\langle X_i, X_j'\rangle/d' - \tau^{-}}{\tau^{+}-\tau^{-}}2 -1  \Big) \cdot d' \cdot (\tau^{+}-\tau^{-})^q  /{2^{3q-1}},
$
where $T_q$ is the degree-$q$ Chebyshev polynomial of the first kind. The algorithm runs in
time $O(d' \cdot n\cdot q)$.
\end{thm}

\begin{framed}
{\tt Inner product ApprxNet}\\

Input: $X =[x_1,\ldots,x_n] $ with each $x_i \in \SS^{d-1}$, $\rho \in [-1,1]$, $\epsilon \in (0,1/2]$.

Output: Sets $C, F \subseteq [n]$.

\begin{itemize}
\item If $\rho\leq \epsilon$, then:
\begin{itemize}
\item $C \gets \emptyset$, $F\gets \emptyset$, $W \gets \{x_1,\ldots,x_n\}$
\item While $W\neq \emptyset$:
\begin{itemize}
\item Choose arbitrary vector $x\in W$.
\item $W \gets W \setminus \{y \in W \mid \langle x,y\rangle \geq \rho-\epsilon \}$
\item $C \gets C \cup \{x\}$
\item If $\forall y \in W$, $\langle x,y \rangle<\rho-\epsilon $ then $F\gets F. \cup \{x\}$
\end{itemize}
\item Return indices of $C$, $F$.
\end{itemize}
\item Apply Theorem \ref{ThmUnif} for input $X$, $\delta=\epsilon/2 \pi$  
and output $Y\in \{-1,1\}^{d' \times n}$ for $d'=O(\log n/\delta^2)$.
\item Apply Theorem \ref{ThmCheb} 
for input $Y$, $d''= n^{0.2}$, $q=50^{-1} \log n$, 
$\tau^-=-1$, $\tau^{+}=1-\frac{2 \cos^{-1}(\rho-\epsilon)}{\pi} +\delta$
and output $Z, Z'$. 
\item Run algorithm {\tt Crude ApprxNet} 
with input $\tau=3n^{0.16}$, $\alpha=\sqrt{\epsilon}/500$, $Z,Z'$ and output $C$, $F$.
\item Return $C$, $F$.
\end{itemize}
\end{framed}

\begin{thm}\label{AppIPnet}
The algorithm {\tt Inner product ApprxNet}, on input  $X =[x_1,\ldots,x_n] $ with each $x_i \in \SS^{d-1}$, $\rho \in [-1,1]$ and $\epsilon \in (0,1/2]$,
computes an approximate $\rho$-net with additive 
error 
$\epsilon$, using the notation of Definition \ref{DfnNetInn}. The algorithm runs in time $\tilde{O}(dn+n^{2-\sqrt{\epsilon}/600})$ and 
succeeds with probability $1-O(1/n^{0.2})$. Additionally, it computes a set $F$ with the following property:
$\{x_i \mid \forall x_j \neq x_i~ \langle x_j,x_i \rangle  < \rho -\epsilon \}\subseteq F \subseteq \{x_i \mid \forall x_j \neq x_i~ \langle x_j,x_i \rangle < \rho  \}$.
\end{thm}

\section{Approximate nets in high dimensions}
\label{SGeneral}
In this section, we translate the problem of computing $r$-nets in $(\RR^d,\|\cdot \|)$ to the problem of computing 
$\rho$-nets for unit vectors under inner product. One intermediate step is that of computing $r$-nets for 
unit vectors under Euclidean distance. 

\subsection{From arbitrary to unit vectors}
In this subsection, we show that if one is interested in 
finding an $r$-net for $(\RR^d,\|\cdot \|)$, it is sufficient 
to solve the problem for points on the unit sphere. One analogous statement is used in \cite{Val15}, where they prove that one can apply a randomized mapping from the general Euclidean space to points on a unit sphere, while preserving the ratio of distances for any two pairs of points. The claim derives by the simple observation that an $r$-net in the initial space can be approximated 
by computing an $\epsilon r/c$-net on the sphere, where $c$ 
is the maximum norm of any given point envisaged as vector. Our exposition is even simpler since we can directly employ the analogous theorem from \cite{Val15}.

\begin{cor}\label{Standardize}
There exists an algorithm, {\tt Standardize}, which, on input a $d \times n$ matrix $X$ with entries $x_{i,j} \in \RR$, a constant $\epsilon \in (0, 1)$ and a distance parameter $r \in \RR$, outputs a $m'\times n $ matrix $Y$, with columns having unit norm and $m'=\log^3n$, and a distance parameter $\rho \in \RR $, such that a $\rho$-net of $Y$ is an approximate $r$-net of $X$, with probability $1-o(1/poly(n))$.
\end{cor}

\subsection{Approximate nets under Euclidean distance}
In this subsection, we show that one can translate the problem of computing an $r$-net for points on the unit sphere under Euclidean distance, to finding an $r$-net for unit vectors under  inner product as defined in Section \ref{SInner}. Moreover, we  identify the subset of the $r$-net which contains 
the centers that are approximately far from any other point. Formally,

\begin{dfn}
Given a set of points $X$ and $\epsilon>0$, 
a set $F\subseteq X$ of $(1+\epsilon)$-approximate $r$-far points is defined by the following property:  $\{x\in X \mid \forall x \neq y \in X ~ \|x-y\| > (1+\epsilon)r \}\subseteq F \subseteq \{x\in X \mid \forall x \neq y \in X ~ \|x-y\| > r \}$.
\end{dfn}

If $r$ is greater than some constant, the problem can be immediately solved by the law of cosines. If $r$ cannot be considered as constant, we distinguish cases $r\geq 1/n^{0.9}$ and $r <1/n^{0.9}$. The first case 
is solved by a simple modification of an analogous algorithm in \cite[p.13:28]{Val15}. The second case is not straightforward and requires partitioning the pointset in a manner which allows computing $r$-nets for each part separately. Each part has bounded diameter which implies that we need to solve a ``large $r$" subproblem.

\begin{thm}\label{ThmLargeRadius}
There exists an algorithm, {\tt ApprxNet(Large radius)}, which, 
for any constant $\epsilon\in (0,1/2]$, $X\subset \SS^{d-1}$ s.t. $|X|=n$, outputs
a $(1+\epsilon)r$-net and a set of $(1+\epsilon)$-approximate $r$-far points with probability $1-O(1/n^{0.2})$. Additionally, 
provided $r>1/n^{0.9}$ the runtime of the algorithm is $\tilde{O}(dn^{2-\Theta(\sqrt[]{\epsilon})})$. 
\end{thm}

Let us now present an algorithm which translates the problem of finding an $r$-net for $r<1/n^{0.9}$ to the problem of computing an $r$-net for $r\geq 1/n^{0.9}$. The main idea is that we compute disjoint subsets $S_i$, which are far enough from each other, so that we can compute $r$-nets for each $S_i$ independently. We show that for each $S_i$ we can compute 
$T_i \subseteq S_i$ which has bounded diameter and 
$T_i'\subseteq S_i$ such that $T_i$, $T_i'$ are disjoint, each point in $T_i$ is far from each point in $T_i'$, and $|T_i'|\leq 3|S_i|/4$. It is then easy to find $r$-nets for $T_i$ by employing the ApprxNet(Large radius) algorithm. Then, we recurse on $T_i'$ which contains a constant fraction of points from $|S_i|$. Then, we cover points in 
$S_i \setminus(T_i \cup T_i')$ and points which do not belong to any $S_i$.

\begin{framed}
{\tt ApprxNet(Small radius)}\\

  Input: $X =[x_1,\ldots,x_n]^T $ with each $x_i \in \SS^{d-1}$, $r< 1/n^{0.9}$, $\epsilon \in (0,1/2]$.

Output: Sets $R, F \subseteq [n]$.

\begin{enumerate}
\item Project points on a uniform random unit vector 
and consider projections $p_1,\ldots,p_n$ which 
wlog correspond to $x_1,\ldots,x_n\in \RR^d$.
\item 
Traverse the list as follows:
\begin{itemize}
\item If $|\{j  \mid p_j \in [p_i-r,p_i] \}| \leq n^{0.6}$ or $i=n$:
\begin{itemize}
\item If $|\{j  \mid  p_j <p_i  \}| \leq n^{0.9}$ remove 
from the list all points $p_j$ s.t. $p_j<p_i-r$ and save set $K=\{x_j \mid p_j\in [p_i-r,p_i] \}$.
\item If $|\{j  \mid p_j <p_i\}| > n^{0.9}$ save sets $K_i=\{x_j \mid p_j\in [p_i-r,p_i] \} \cup K$,
$S_i=\{x_j\mid p_j<p_i-r\}\setminus K$ and remove projections of $S_i$ and $K_i$ from the list.
\end{itemize}
\end{itemize}
\item After traversing the list if we have not saved any $S_i$ go to 5; otherwise for each $S_i$:
\begin{itemize}
\item For each $u\in S_i$, sample 
$n^{0.1}$ distances between $u$ and randomly chosen 
$x_k\in S_i$. Stop if for the selected $u\in S_i$, more than $1/3$ of the sampled points are in distance $\leq  r n^{0.6}$. This means that 
one has found $u$ s.t. 
$|\{x_k \in S_i, \|u-x_k\|\leq  r n^{0.6}\}| \geq |S_i|/4$ with high probability. If no such point was found, output "ERROR".
\item Let $0\leq d_1\leq \ldots\leq d_{|S_i|}$ be the 
distances between $u$ and all other points in $S_i$. 
Find $c\in[ r n^{0.6},2r  n^{0.6}]$ s.t. 
$|\{j \in [n] \mid d_j \in [c,c+r] \}| <n^{0.4}$,
 store $W_i=\{x_j \mid d_j \in [c,c+r] \}$, 
 and remove $W_i$ from $S_i$. 
\item Construct the sets $T_i=\{x_j \in S_i \mid d_j<c\}$ 
and $T_i'=\{x_j \in S_i \mid d_j > c+r\}$. 
\begin{itemize}
\item For $T_i$, subtract $u$ from all vectors in $T_i$, run {\tt Standardize}, then {\tt ApprxNet (Large radius)}, both with {$\epsilon/4$}. Save points which correspond to output at $R_i$, $F_i$ respectively.
\item Recurse on $T_i'$ the whole algorithm, and notice that $|T_i'|\leq 3 |S_i|/4$. Save output at $R_i'$, and $F_i'$ respectively.
\end{itemize}
\end{itemize}
\item 
Let $R \gets \bigcup_i R_i \cup R_i'$ and 
$F \gets \bigcup_i F_i \cup F_i'$. Return to the list $p_1,\ldots,p_n$.
\begin{enumerate}
\item {Remove from $F$ all points which cover at least one point from $\bigcup_i W_i$ or $\bigcup_i K_i$.}
\item \label{itm:deleteb}
Delete all points $ (\bigcup_i T_i) \setminus (\bigcup_i R_i)$, and $ (\bigcup_i T_i') \setminus (\bigcup_i R_i')$.
\item \label{itm:deletec}For each $i$ delete all points in $W_i$ covered by $R_i$, or covered by $R_i'$. 
\item \label{itm:deleted}For each $i$ delete all points in $K_i$ covered by $R$.
\item Finally delete $R$ from the list. Store the remaining points at $F'$.
\end{enumerate}
\item $R' \gets \emptyset$. Traverse the list as follows: For each $p_i$, check the distances from all $x_j$ s.t. 
$p_j\in [p_i-r,p_i]$.
\begin{itemize}
\item If $\exists\, x_j \in R' :$ $\|x_i-x_j\| \leq r$, delete $x_i$ from the list, set $F' \gets F' \backslash \{x_i , x_j\}$  and continue traversing the list. 
\item If there is no such point $x_j$ then $R \gets R \cup \{x_i\}$ and continue traversing the list.
\end{itemize}
\item Output indices of $R\gets R \cup R'$ and $F \gets F \cup F'$.
\end{enumerate}
\end{framed}

\begin{thm}\label{ThmSmallr}
For any constant $\epsilon>0$, $X\subset \SS^{d-1}$ s.t. 
$|X|=n$, and $r < 1/n^{0.9}$, {\tt ApprxNet(Small radius)} will output
a $(1+\epsilon)r$-net and a set of $(1+\epsilon)$-approximate $r$-far points in time $\tilde{O}(dn^{2-\Theta(\sqrt[]{\epsilon})})$, with probability $1-o(1/n^{0.04})$.
\end{thm}

\begin{proof}
Note that points in $S_i$ had projections
$p_i$ in sets of contiguous intervals of width $r$; each interval had $\geq n^{0.6}$ points,
hence the diameter of the projection of $S_i$ is $\leq n^{0.4}r$.
By the Johnson Lindenstrauss Lemma \cite{DG02} we have that for $v \in \SS^{d-1}$ chosen uniformly at random:
$$\Pr\Big[\langle u,v\rangle^{2}\leq \frac{\|u\|^2}{n^{0.4}}\Big]\leq \frac{\sqrt{d} \sqrt{e}}{n^{0.2}}.
$$
Hence,
$
\EE[| \{ x_k,x_j\in S_i \mid \|x_k-x_j\| \geq n^{0.6}r \text{ and }  \|p_k-p_j\|\leq n^{0.4} r\}|]
\leq |S_i|^2 \cdot  \frac{\sqrt{e d}}{n^{0.2}},
$
and the probability
$$
\Pr[| \{ x_k,x_j\in S_i \mid \|x_k-x_j\| \geq n^{0.6}r \text{ and }  \|p_k-p_j\|\leq n^{0.4} r\}| \geq |S_i|^{1.95}] \leq |S_i|^{0.05} \cdot  \frac{\sqrt{e d}}{n^{0.2}}\leq \frac{\sqrt{e d}}{n^{0.15}}.
$$
Taking a union bound over all sets $S_i$ yields a  probability of failure $o({1}/{n^{0.045}})$. 
This implies that (for large enough $n$, which implies large enough $|S_i|$) at least 
$${\binom{|S_i|}{2}} -|S_i|^{1.95}\geq {\frac{|S_i|^2}{4}}$$ 
distances between points in $S_i$ are indeed small 
($\leq n^{0.6}r$). Hence, there exists some point $p_k \in S_i$ which $(n^{0.6}r)$-covers $|S_i|/2$ points. For each possible $p_k$ 
we sample $n^{0.1}$ distances to other points, and by {Chernoff bounds}, if a point $(n^{0.6}r)$-covers a fraction of more than $1/2$ of the points in $S_i$, then it covers more than 
$n^{0.1}/3$ sampled points with high probability. Similarly, if 
a point $(n^{0.6}r)$-covers a fraction of less than $1/4$ of the points in $S_i$, then it covers less than 
$n^{0.1}/3$ sampled points with high probability.
More precisely, for some fixed $u\in S_i$, let $X_j=1$ 
when for the $j$th randomly chosen point $v\in S_i$, it holds 
$\| u-v\| \leq  n^{0.6}r$ and let $X_j=0$ otherwise. Then, for $Y=\sum_{j=1}^{n^{0.1}} X_j$, it holds:
$$
\EE[Y]\geq n^{0.1}/2 \implies \Pr[Y\leq n^{0.1}/3 ]\leq \exp(- \Theta(n^{0.1})),
$$
$$
\EE[Y]\leq n^{0.1}/4 \implies \Pr[Y\geq n^{0.1}/3]\leq \exp(- \Theta(n^{0.1})).
$$
Since for any point  $x\in T_i$ and any point $y \in T_i'$ we have $\|x-y\|>r$, the packing property of $r$-nets is preserved when we build $r$-nets for $T_i$ and $T_i'$ independently. For each $T_i$, we succeed in building $r$-nets with probability $1-O(1/n^{0.2})$. By a union bound over all sets $T_i$, we have a probability of failure $O(1/n^{0.1})$. 
Furthermore, points which belong to sets $W_i$ and $K_i$ are possibly covered and need to be checked.

For the analysis of the runtime of the algorithm, notice that step \ref{itm:deleteb} costs time 
$O(d\cdot (\sum_i|T_i|+\sum_i|T_i'|))=O(dn)$. Then, 
step \ref{itm:deletec} costs time $O(d \cdot \sum_i |W_i|\cdot |T_i|+d \cdot \sum_i |W_i|\cdot |T_i'|)=O(d n^{1.4})$. Finally, notice that we have at most $n^{0.1}$ sets $K_i$. Each $K_i$ contains at most $2n^{0.6}$ points, hence checking each point in $\bigcup_i K_i$ with each point in $R$ costs $O(d n^{1.7})$.

Now regarding step 5, consider any interval $[p_i-r,p_i]$ in the initial list, where all points are projected. If $|\{ j \mid p_j \in [p_i-r,p_i]\}\leq 2 n^{0.9}$ then the $i$th iteration in step 5 will obviously cost $O(n^{0.9})$, since previous steps only delete points. If $|\{ j \mid p_j \in [p_i-r,p_i]\}> 2 n^{0.9}$, we claim that 
$|\{j<i \mid  p_j \in [p_i-r,p_i] \text{ and } K_j \text{ is created}\}| \leq 1$. Consider the smallest $j <i$ s.t. $K_j$ is created and $p_j\in [p_i-r,p_i]$. This means that all points $p_k$, for $k\leq j$, are deleted when $p_j$ is visited. Now assume that there exists integer $l \in (j,i)$ s.t. $K_l$ is created. This means that the remaining points in the interval $[p_l-r,p_l]$ are  
$\leq n^{0.6}$ and all of the remaining points $p_k <p_l$ are more than $n^{0.9}$. This leads to contradiction, since by the deletion in the $j$th iteration, we know that all of the remaining points $p_k <p_l$ lie in the interval $[p_l-r,p_l]$. 

Now, assume that there exists one $ j<i$ s.t. $p_j \in [p_i-r,p_i]$ and $K_j$ is created. Then, when $p_i$ is visited, there at least $2 n^{0.9}-n^{0.6}>n^{0.9}$ remaining points in the interval $[p_i-r,p_i]$. Hence, there exists $l\geq i$  for which 
the remaining points in the interval $[p_i-r,p_i]$ 
are contained in $S_l \cup K_l$. Hence in this case, in step 5, there exist at most $O(n^{0.6})$ points which are not deleted and belong to the interval $[p_i-r,p_i]$. Now assume that there does not exist any $ j<i$ s.t. $p_j \in [p_i-r,p_i]$ and $K_j$ is created. This directly implies that  there exists $l\geq i$  for which 
the remaining points in the interval $[p_i-r,p_i]$ 
are contained in $S_l \cup K_l$.

At last, the total time of the above algorithm is dominated by the calls to the construction of the partial $r$-nets of the sets $T_i$. Thus, the total running time is $O(\sum_{ i} {|T_i|}^{2-\Theta(\sqrt{\epsilon})}+\sum_{ i} {|T_i|'}^{2-\Theta(\sqrt{\epsilon})})=
O(\sum_{ i} {|T_i|}^{2-\Theta(\sqrt{\epsilon}})+\sum_{i} {(3|T_i|/4)}^{2-\Theta(\sqrt{\epsilon})})=
\tilde{O}(n^{2-\Theta(\sqrt{\epsilon}))}).$
{Finally, taking a union bound over all recursive calls of the algorithm we obtain a probability of failure $o(1/n^{0.04})$.} 
\end{proof}

We now present an algorithm for an 
$(1+\epsilon)r$-net for points in $\RR^d$ under Euclidean distance. 

\begin{framed}
{\tt ApprxNet}\\

Input: Matrix $X=[x_1,\ldots,x_n]$ with each $x_{i} \in \RR^d$, parameter $r \in \RR$, constant $\epsilon \in (0, 1/2]$.

Output: $R \subseteq \{x_1,\ldots,x_n\}$ 

\begin{itemize}
\item  Let $Y$, $r'$ be the output of algorithm {\tt Standardize} on input $X$, $r$ with parameter $\epsilon/4$.
\item If $r \geq 1/n^{0.9}$ run {\tt ApprxNet(Large radius)} on input $Y$,  $\epsilon/4, r'$ and return points which correspond to  the set $R$.
\item If $r < 1/n^{0.9}$ run {\tt ApprxNet(Small radius)} on input $Y$,  $\epsilon/4, r'$ and return points which correspond to the set $R$.
\end{itemize}
\end{framed}

\begin{thm}\label{ApprxNet}
Given $n$ points in $\RR^d$, a distance parameter $r \in \RR$ and an approximation parameter $\epsilon \in (0, 1/2]$, with probability $1-o(1/n^{0.04})$, {\tt ApprxNet} will return a $(1+\epsilon)r-net$, $R$, in $\tilde{O}(dn^{2-\Theta(\sqrt{\epsilon})})$ time.
\end{thm}
\begin{proof}
The theorem is a direct implication of Theorems \ref{ThmLargeRadius}, \ref{ThmSmallr}, \ref{Stand}.
\end{proof}

\begin{thm}\label{DelFar}
Given $X\subset\RR^d$ such that $|X|=n$, a distance parameter $r \in \RR$ and an approximation parameter $\epsilon \in (0, 1/2]$, there exists an algorithm, {\tt DelFar}, that will return, with probability $1-o(1/n^{0.04})$, a set $F'$ with the following properties in $\tilde{O}(dn^{2-\Theta(\sqrt{\epsilon})})$ time:
\begin{itemize}
\item If for a point $p \in X$ it holds that  $\forall q\neq p, q \in X$ we have $\|p-q\| > (1+\epsilon)r$, then $p \notin F'$.
\item If for a point $p \in X$ it holds that  $\exists q\neq p, q \in X$ s.t. $\|p-q\| \leq r$, then $p \in F'$.
\end{itemize}
\end{thm}

\section{Applications and Future work}\label{Sapps}

Concerning applications, in \cite{HR15}, they design an approximation scheme, which solves various
distance optimization problems. The technique employs a grid-based construction of $r$-nets which is linear in $n$, but exponential in $d$. The main prerequisite of the method is the existence of 
a linear-time decider (formally defined in Appendix~\ref{Aframework}). The framework is especially interesting when the dimension is constant, since the whole algorithm costs time linear in $n$ which, for some problems, improves upon previously known near-linear algorithms. 
When the dimension is high, we aim for polynomial dependency on $d$, and subquadratic dependency on $n$.


Let us focus on the problem of approximating the {\it $k$th nearest neighbor distance}. 
\begin{dfn}
Let $X\subset \RR^d$ be a set of $n$ points, approximation error $\epsilon>0$, and let $d_1\leq \ldots \leq d_n$ be the nearest neighbor distances. The problem of computing an $(1+\epsilon)$-approximation to the {\it $k$th nearest neighbor distance} asks for a pair 
$x,y\in X$ such that $\|x-y\|\in [(1-\epsilon)d_k,(1+\epsilon)d_k]$.
\end{dfn}

Now we present an approximate decider for the problem above. 
This procedure combined with the framework we mentioned earlier, which employs our net construction, results in an efficient solution for this problem in high dimension. 
\begin{framed}
{\tt kth NND Decider}\\

Input: $X \subseteq \RR^d$, constant $\epsilon\in (0,1/2]$, integer $k>0$.

Output: An interval for the optimal value $f(X, k)$.
\begin{itemize}
\item Call {\tt DelFar}$(X, \frac{r}{1+\epsilon/4}, \epsilon/4)$ and store its output in $W_1$.
\item Call {\tt DelFar}$(X, r, \epsilon/4)$ and store its output in $W_2$.
\item Do one of the following:
\begin{itemize}
\item If $|W_1| > k$, then output $``f(X, k) < r"$.
\item If $|W_2| < k$, then output $``f(X, k) > r"$.
\item If $|W_1| \leq k$ and $\abs{W_2} \geq k$, then output $``f(X, k) \in [\frac{r}{1+\epsilon/4}, \frac{1+\epsilon/4}r]"$.
\end{itemize}
\end{itemize}
\end{framed}
 
\begin{thm}\label{KND}
Given a pointset $X \subseteq \RR^d$, one can compute a $(1+\epsilon)$-approximation to the $k$-th nearest neighbor in $\tilde{O}(dn^{2-\Theta(\sqrt{\epsilon})})$, with probability $1-o(1)$.
\end{thm}
To the best of our knowledge, this is the first high dimensional solution for this problem. 
Setting $k=n$ and applying Theorem \ref{KND} one can compute the {\it farthest nearest neighbor} in $\tilde{O}(dn^{2-\Theta(\sqrt{\epsilon})})$ with high probability.

Concerning future work, let us start with the problem of finding a greedy permutation. 
A permutation $\Pi = <\pi_1, \pi_2,\dots >$ of the vertices of a metric space $(X, \norm{\cdot})$ is a \textit{greedy permutation}  if each vertex $\pi_i$ is the farthest in $X$ from the set $\Pi_{i-1} = <{\pi_1,\dots, \pi_{i-1}}>$ of preceding vertices. The computation of $r$-nets is closely related to that of the greedy permutation.

The $k$-center clustering problem asks the following: given a set $X \subseteq \RR^d$ and an integer $k$, find the smallest radius $r$ such that $X$ is contained within $k$ balls of radius $r$.
By \cite{EHS15}, a simple modification of our net construction implies an algorithm for the $(1+\epsilon)$ approximate greedy permutation in time $\tilde{O}(d n^{2-\Theta(\sqrt{\epsilon})} \log \Phi)$ where $\Phi$ denotes the spread of the pointset. 
Then, approximating the greedy permutation implies a 
$(2+\epsilon)$ approximation algorithm for $k$-center clustering problem. We expect that one can avoid any  dependencies on $\Phi$.

\subsection*{Acknowledgment.}
I.Z.~Emiris acknowledges partial support by the EU
H2020 research and innovation programme, under the Marie Sklodowska-Curie grant
agreement No 675789: Network ``ARCADES".

\newpage 
\bibliographystyle{alpha}
\bibliography{nets}

\newpage
\appendix

\section{Proof of Theorem \ref{ThmVecAgg}}

\begin{lem}[Anti-concentration]\label{LemAntiCon}
Let $q_1,\ldots,q_t \in \{-1,1\}$ be chosen independently and uniformly at random, and let $a_1,\ldots,a_t\in \RR$ s.t. $|a_1|=\max_i |a_i|$. Then,
$$\Pr[|\sum_{i=1}^t q_i \cdot a_i| \geq |a_1|]\geq 1/2.$$
 
\end{lem}
\begin{proof}
 Consider a given assignment for $q_2,\ldots,q_t$. Then if 
 $$\sum_{i=2}^t q_i \cdot a_i =0 \implies | \sum_{i=1}^t q_i \cdot a_i|= |q_1 \cdot a_1|=|a_1|.$$
 Otherwise,
 $$\Pr[|\sum_{i=1}^t q_i \cdot a_i| \geq |a_1|] \geq \Pr[sign(q_1 \cdot a_1)=sign(\sum_{i=2}^t q_i \cdot a_i =0)]=1/2.
 $$
\end{proof}

\noindent{\textbf{Proof of Theorem \ref{ThmVecAgg}.}} 
Notice that
 $$w_{j,k}^i=\sum_{x_i\in S_k} q_i \cdot \langle x_j, x_i \rangle$$
 and since $q_1,\ldots,q_{|S_k|}\in \{-1,1\}$ are independent and chosen uniformly at random, we obtain  
 $$\EE [w_{j,k}^i ]=0.
 $$
 
 If $\forall u \in S_k$, $ |\langle x_j, u \rangle|\leq \tau$, then 
 $$Var(w_{j,k}^i)=\EE [(w_{j,k}^i)^2]\leq n^{2\alpha} \tau^2$$
 By Chebyshev's inequality:
 $$\Pr[|w_{j,k}^i|\geq 3 \cdot n^{\alpha} \tau]\leq 1/9$$
 
With $m$ repetitions, the number of successes $N$, 
that is the number of indices $i$ for which 
$|w_{j,k}^i|\leq 3 \cdot n^{\alpha} \tau$, follows the binomial distribution. Hence, 
 
$$\Pr[N \leq 3m/4] \leq exp( -m/26).$$

We consider as bad event the event that for some $j,k$, more than $25\%$ of the repetitions fail, that is $|w_{j,k}^i|\geq 3 \cdot n^{\alpha} \tau$. 
By the union bound, this probability is $\leq n^{2-\alpha} \cdot exp(- m/26)$, which for $m\geq78 \log n$ implies a probability of failure $\leq 1/n^3$.  

Now consider $x_j$, and $x_l\in S_k$ s.t. $| \langle x_j,x_l\rangle| \geq 3 \cdot n^{\alpha} \tau$, then by Lemma~\ref{LemAntiCon}, 
with probability $1/2$, $|w_{j,k}^i| \geq  3 \cdot n^{\alpha} \tau$. 
We consider as bad event the event that for $j,l$, more than $75\%$ of the repetitions fail, that is $|w_{j,k}^i|\leq 3 \cdot n^{\alpha} \tau$.
Hence, 
$$\Pr[N \leq m/4] \leq exp(-m/8 ),$$
which for $m\geq78 \log n$ implies a probability of failure $\leq 1/n^3$.

The runtime of the algorithm is dominated, up to polylogarithmic factors, by the computation of matrix $Z$, taking time $O(dn)$, the computation 
of matrix $W$, taking time $n^{2-a}$, or the 
computation of the product $W^i$, taking time 
$MatrixMul(n \times d, d \times n^{1-a})$.

\hfill \qed
\section{Proof of Theorem \ref{ThmCheb}}\label{AppThmCheb}
We refer to \cite[Algorithm 3: Chebyshev Embedding]{Val15}. The proof is the same with that of \cite{Val15}, apart from indicating that the probability of success is actually $1-o(1/n)$ instead of $1-o(1)$ as stated in \cite{Val15}. While $1-o(1/n)$ probability of success is enough for our purposes, even better probability bounds can be achieved.

The fact that all inner products are concentrated within
$\pm \sqrt{m} \log n$ about their expectations follows from the fact that each row of $Y$, $Y'$ is
generated identically and independently from the other rows, and all entries of these
matrices are $\pm 1$; thus, each inner product is a sum of independent and identically
distributed random $\pm 1$ random variables, and we can apply the basic Chernoff bound
to each inner product, and then a union bound over the $O(n^2 )$ inner products. Let $X_i \in \pm 1$ i.i.d. random variables. The basic chernoff bound gives probability,
$$
\Pr[ \, |\sum_{i=1}^{m'} X_i - \EE[\sum_{i=1}^{m'} X_i ]| >\sqrt{m'}\log n ]\leq 2 \cdot exp(-\Theta(\log^2 n))=o(1/n^3).
$$

Given 
this concentration, we now analyze the expectation of the inner products. Let $u, u'$ be
columns of $X, X'$ , respectively, and $v, v'$ the corresponding columns of $Y, Y'$. Letting 
$x=\langle u,u'\rangle/m$, we argue that by \cite[Lemma 3.3]{Val15}, $\EE[ v, v' ] = m' \sum_{i=1}^q \frac{x-c_i}{2}$ (1), 
where $c_i$ is the location of the $i$th root of the $q$th Chebyshev polynomial after the roots
have been scaled to lie in the interval $[\tau^{-} , \tau^+ ]$. To see why this is the case, note that
each coordinate of $u$, $u'$ ,is generated by computing the product of $q$ random variables 
that are all $\pm 1$; namely, a given entry of $u$ is given by $\prod_{l=1}^q s_v(l)$, with the corresponding 
entry of $u'$ given by
$\prod_{l=1}^q s_{v'}(l)$. Note that for $i \neq j$, $s_v (i)$ is independent of $s_v ( j)$ and
$t_{v'} ( j)$, although by construction, $s_v (i)$ and $t_{v'} (i)$ are not independent. We now argue that 
$\EE[s_v(i)t_{v'}(i)] = \frac{x-c_i}{2}$
, from which Eq. (1) will follow by the fact that the expectation of the 
product of independent random variables is the product of their expectations.

By construction, in Step (1) of the inner loop of the algorithm, with probability $1/2$,
$E[s_v (i)t_{v'} (i)] = \langle v, v'\rangle /m = x$. Steps (2)–(4) ensure that with the remaining $1/2$  probability, $\EE[s_v (i)t_v (i)] = \frac{1-c_i}{2}
(1) - \frac{1+c_i}{2}
(-1) = -c_i .$ Hence, in aggregate over the randomness of 
Steps (1)–(4), $\EE[s_v(i)t_{v'} (i)] = x/2 - c_i/ 2 i$ , as claimed, establishing Eq.~(1).

To show that Eq.~(1) yields the statement of the proposition, we simply reexpress 
the polynomial $\prod_{i=1}^q \frac{x-c_i}{2}$
in terms of the $q$th Chebyshev polynomial $T_q$. Note that the 
$q$th Chebyshev polynomial has leading coefficient $2^{q-1}$, whereas this expression (as a
polynomial in $x$) has leading coefficient $1/2^q$,  disregarding the factor of the dimension
$m'$. If one has two monic degree $q$ polynomials, $P$ and $Q$ where the roots of $Q$ are given
by scaling the roots of $P$ by a factor of $\alpha$, then the values at corresponding locations
differ by a multiplicative factor of $1/\alpha^q$; since the roots of $T_q$ lie between $[-1, 1]$ and
the roots of the polynomial constructed in the embedding lie between $[\tau^{-} , \tau^{+} ]$, this 
corresponds to taking $\alpha = \frac{2}{\tau^{+}-\tau^{-}}$.

\section{Proof of Theorem \ref{AppIPnet}}

\begin{thm}[\cite{Cop97}]\label{ThmCop}
For any positive $\gamma > 0$, provided that $\beta < 0.29$, the product of a $k \times k^\beta$ with
a $k^\beta \times k$ matrix can be computed in time $O(k^{2+\gamma} )$.
\end{thm}

\begin{cor}\label{CorMatrMul}
For any positive $\gamma > 0$, provided that $\beta < 0.29 \cdot \alpha <1 $, the product of a $n \times n^\beta$ by
a $n^\beta \times n^\alpha$ matrix can be computed in time $O(n^{1+\alpha+\alpha \gamma})$.
\end{cor}
\begin{proof}
The idea is 
to perform $n^{1-\alpha}$ multiplications of 
matrices of size $n^{\alpha} \times n^{\beta}$ 
and $n^{\beta} \times n^{\alpha}$.

Hence, by Theorem~\ref{ThmCop}, the total cost is:
$$
O(n^{1-\alpha}(n^{\alpha(2+\gamma)}))=O(n^{1+\alpha+\alpha \gamma}).
$$
\end{proof}
\begin{fct}\label{FactCheb}
Let $T_q(x)$ denote the $q$th Chebyshev polynomial 
of the first kind, then the following hold:
\begin{itemize}
\item For $x\in [-1,1]$, $|T_q(x)|\leq 1$.
\item For $\delta \in (0,1/2]$, $T_q(1+\delta)\geq \frac{1}{2} e^{q \sqrt{\delta}}$.
\end{itemize}
\end{fct}

\begin{clm}\label{ClmArccos}
For $\rho\in[-1,1]$, $\epsilon\in (0,1)$, it holds
 ${ \cos^{-1}(\rho-\epsilon)}-{ \cos^{-1}(\rho)}\geq \epsilon/2$.
\end{clm}
\begin{proof}
If $(\rho -\epsilon)^2 \neq 1$ then we have
$${ \cos^{-1}(\rho-\epsilon)}-{ \cos^{-1}(\rho)}=
\int_{\rho -\epsilon}^1\frac{1}{\sqrt{1-x^2}} \mathrm{d}x - \int_{\rho }^1\frac{1}{\sqrt{1-x^2}} \mathrm{d}x =
$$
$$
= \int_{\rho-\epsilon }^{\rho}\frac{1}{\sqrt{1-x^2}} \mathrm{d}x=\int_{0}^\epsilon \frac{1}{\sqrt{1-(\rho-\epsilon+y)^2}} \mathrm{d}y  
\geq \int_{0}^\epsilon \frac{1}{\sqrt{1-(\rho-\epsilon)^2}} \mathrm{d}y=\frac{\epsilon}{\sqrt{1-(\rho-\epsilon)^2}}\geq \epsilon.
$$

Now if $(\rho-\epsilon)^2\neq 1 \implies \rho-\epsilon=-1$ then, 
$$
{ \cos^{-1}(\rho-\epsilon)}-{ \cos^{-1}(\rho)}=
\int_{-1 }^{-1+\epsilon}\frac{1}{\sqrt{1-x^2}} \mathrm{d}x  \geq \frac{\epsilon}{\sqrt{2\epsilon-\epsilon^2}}\geq \epsilon/2.
$$
\end{proof}

\noindent{\textbf{Proof of Theorem  \ref{AppIPnet}.}} 
If $\rho \leq \epsilon$, our approach ensures that 
for any $x,y\in C$, it holds $\langle x, y \rangle< \rho-\epsilon \leq 0$. We show that $|C|\leq d+1$, 
due to a simple packing argument. Let 
$x_1,\ldots,x_{d+2}$ such that $\forall i\neq j \in [d+2]$
we have $\langle x_i, x_j \rangle <0$. Then, there exist $\lambda_1,\ldots,\lambda_{d+1} \in \RR$ not all zero for which 
$\sum_{i=1}^{d+1} \lambda_i x_i =0$. 
Now consider two subsets $I,J \subseteq [$d$+2]$ of 
indices such that $\forall i \in I, \lambda_i > 0$ 
and $\forall j \in J, \lambda_j < 0$. We can write $\sum_{i\in I} \lambda_i x_i=\sum_{j\in J} -\lambda_j x_j \implies 
0 \leq \langle \sum_{i\in I} \lambda_i x_i , -\sum_{j\in J} \lambda_j x_j \rangle = -\sum_{i\in I,j \in J} \lambda_i \lambda_j \langle x_i, x_j\rangle <0 $
which leads to contradiction. If $J= \emptyset$ (or equivalently if $I= \emptyset$), then
$0=\langle x_{d+2}, \sum_{i\in I} \lambda_i x_i \rangle <0$, which leads again to contradiction. 

We now focus on the case $\rho>\epsilon$.
By Theorem \ref{ThmUnif}, with probability $1 - o(1/n^2)$, the matrix $Y$ returned by the corresponding 
algorithm will have the property that any pair of columns 

$$\langle X_i ,X_j \rangle\geq \rho 
\implies \frac{\langle Y_i, Y_j \rangle}{d'} \geq 
1-\frac{2 \cos^{-1}(\rho)}{\pi} -\delta$$

$$\langle X_i ,X_j \rangle\leq \rho-\epsilon 
\implies \frac{\langle Y_i, Y_j \rangle}{d'} \leq 
1-\frac{2 \cos^{-1}(\rho-\epsilon)}{\pi} +\delta.$$
Hence, according to 
Claim \ref{ClmArccos}, it suffices to set $\delta=\epsilon/3 \pi$ in order to distinguish between the two cases:
$$
1-\frac{2 \cos^{-1}(\rho-\epsilon)}{\pi} +2\delta \leq 
1-\frac{2 \cos^{-1}(\rho)}{\pi} -\delta.
$$

Now we set $\tau^{+}=1-\frac{2 \cos^{-1}(\rho-\epsilon)}{\pi} +\delta>-1$. 
By Theorem \ref{ThmCheb}, with probability $1-o(1)$, 
$$\langle Y_i,Y_j\rangle\leq \tau^{+} d' \leq \implies|\langle Z_i, Z_j \rangle| \leq d'' \frac{2^{q}}{2^{3q-1}}+\sqrt{d''} \log n \leq 3n^{0.16}
$$
for large enough $n$.
Moreover, let $Y_i,Y_j$ s.t. 
$\langle Y_i, Y_j \rangle \geq  (\tau^{+}+\delta) d'$. Then, 
$$|\langle Z_i, Z_j' \rangle| \geq 
 d'' \cdot T_q \Big(1+2 \frac{\delta}{\tau^{+}+1}\Big) \frac{2^{q}}{2^{3q-1}}-\sqrt{d''} \log n  > \frac{1}{2}\cdot T_q \Big(1+2 \frac{\delta}{\tau^{+}+1}\Big) \cdot n^{0.16}
$$
for large enough $n$.

Then, by Fact \ref{FactCheb}, 
$$|\langle Z_i, Z_j' \rangle| \cdot n^{-0.16}\geq \frac{1}{4} e^{q\sqrt{\delta}} 
= \frac{1}{4} n^{\sqrt{\delta}/50} 
\geq 3n^{\sqrt{\delta}/100} 
\geq 3n^{\sqrt{\epsilon}/400},$$
where some of the inequalities hold for large enough $n$. 

Now, by Theorems \ref{ThmUnif}, \ref{ThmCheb}, \ref{ThmCrudeNet} and Corollary \ref{CorMatrMul} 
the time complexity is $\tilde{O}(dn+n^{2-\sqrt{\epsilon}/600})$, if we set as $\gamma$ in Corollary \ref{CorMatrMul} a sufficiently small multiple of $\sqrt{\epsilon}$. Finally,, the subroutine with the higher probability of failure is {\tt Crude ApprxNet} and by the union bound, it dominates the total probability of failure.
\hfill \qed

\section{Proof of Corollary \ref{Standardize}}
We use an algorithm introduced in \cite{Val15}: its guarantees are stated below. 
\begin{thm}\label{Stand} \cite{Val15}
There exists an algorithm which on input a $d \times n$ matrix $X$ with entries $x_{i,j} \in \RR$ and a constant $\epsilon \in (0, 1)$
outputs a  $m'\times n $ matrix $Y$ with columns having unit norm and $m'=\log^3n$, such that, with probability $1-o(1/poly(n))$ for all sets of four columns $Y_1, Y_2, Y_3, Y_4$ of matrix $Y$, with $X_1, X_2, X_3, X_4$ being the corresponding columns of matrix $X$, it holds that \[ 
\frac{\|Y_1-Y_2\|}{\|Y_3-Y_4\|} \frac{\|X_3-X_4\|}{\|X_1-X_2\|} \in [1-\frac{\epsilon}{10}, 1+\frac{\epsilon}{10} ] .
\]
\end{thm}

Now, let us define two $d$-dimensional vectors $X_{n+1}, X_{n+2}$, s.t.\ $r'=X_{n+1}-X_{n+2}$ and $\|r'\|=r$, where $X$ is a $d \times n$ matrix  with entries $x_{i,j} \in \RR$ and $r \in \RR$ is the radius of the $r$-net of $X$. Also, let matrix $X'$ denote the concatenation of $X$, $X_{n+1}$ and $X_{n+2}$ with size $d\times (n+2)$. After applying Theorem \ref{Stand} on input $X'$ and $\epsilon/10$, we define $\rho:=\|Y_{n+1}-Y_{n+2}\|$ to be the new radius of $Y$. Then,  we claim that the following hold with probability $1-o(1/poly(n))$, which immediately implies Corollary \ref{Standardize}:
\begin{itemize}
\item For all $X_i, X_j \in X$ and their corresponding $Y_i, Y_j \in Y$, if $\|X_i-X_j\| \leq r$ then $\|Y_i-Y_j\| \leq (1+\epsilon/10) \rho$.
\item For all $X_i, X_j \in X$ and their corresponding $Y_i, Y_j \in Y$, if $\|X_i-X_j\| \geq (1+\epsilon)r$ then $\|Y_i-Y_j\| \geq (1+\epsilon/2) \rho$.
\end{itemize}

\noindent{\textbf{Proof of Corollary \ref{Standardize}}}.
From Theorem \ref{Stand}, we easily derive that for all $X_i, X_j \in X$ and their corresponding $Y_i, Y_j \in Y$, it holds that \[\|Y_i-Y_j\| \leq  (1+\epsilon/10) \frac{\|X_i-X_j\|}{r}  \rho\]
Therefore, if $\|X_i-X_j\| \leq r$, we have $\|Y_i-Y_j\| \leq  (1+\epsilon/10) \rho$. For the other direction, we use the opposite side of Theorem \ref{Stand}, thus we have that for all $X_i, X_j \in X$ and their corresponding $Y_i, Y_j \in Y$:
\[
\|Y_i-Y_j\| \geq  (1-\epsilon/10) \frac{\|X_i-X_j\|}{r}  \rho .
\] 
It follows that $\|X_i-X_j\| \geq (1+\epsilon)r \Rightarrow \|Y_i-Y_j\| \geq  (1-\epsilon/10) (1+\epsilon)  \rho \Rightarrow \|Y_i-Y_j\| \geq  (1+\epsilon/2) \rho$.
\hfill \qed

\section{Proof of Theorem \ref{ThmLargeRadius}}

\begin{framed}
{\tt ApprxNet(Large radius)}\\

Input: $X =[x_1,\ldots,x_n]^T $ with each $x_i \in \SS^{d-1}$ with $d=\log^3n$, $r> 1/n^{0.9}$, $\epsilon \in (0,1/2]$.

Output: Sets $R, F\subseteq [n]$.
\begin{itemize}
\item If $r>0.2$ run  {\tt Inner Product ApprxNet} with error parameter $\epsilon/25$ and $\rho=1-\frac{r^2}{2}$.
\item Otherwise, define the $d \times n$ matrix $Z$ as follows:
for each $i \in [d]$, select $q=\floor*{\frac{\pi}{2 \cos^{-1}(1-r^2/2)}}$ uniformly random vectors
$v_1, \ldots, v_q$ and for all $j \in [n]$, set
\begin{center} 
$z_{i,j}=sign
\prod\limits_{k=1}^{k=q}X_j^Tv_k$,
\end{center}
where $X_j$ is the $j$th column of matrix $X$.
\item Run {\tt Inner Product ApprxNet} with $\rho= \Big(1-\frac{2\cos^{-1}(1-r^2/2)}{\pi}\Big)^q$, error parameter $\epsilon/100$ and input matrix $Z$ with all entries 
scaled by $1/\sqrt[]{d}$ to make them have unit norm.
\end{itemize}
\end{framed}

\noindent{\textbf{Proof of Theorem \ref{ThmLargeRadius}.}}

In the case of $r>0.2$ we will show that the $1+\epsilon$ multiplicative approximation on the distance translates to $c \epsilon$ additive approximation to the inner product. Applying the law of cosines, the first condition yields $\langle p,q \rangle \geq 1-\frac{r^2}{2}$ and  the second condition yields
$\langle p,q \rangle \leq 1- \frac{r^2}{2}-\frac{2\epsilon r^2+ (\epsilon r)^2}{2} < 1- \frac{r^2}{2}-\frac{\epsilon}{25}$. So, it suffices 
to take $c=1/25$.

Now suppose that $r<0.2$. For each random vector $v$ we have 
that $\EE[sign(X_i^Tv \cdot X_j^T v)]= 1-\frac{2\theta(X_i,X_j)}{\pi}$, where $\theta(X_i,X_j)$ denotes the angle between $X_i,X_j$. Since expectations of independent random variables multiply, we have that, for each $k$,
\begin{center}
$\EE[z_{k,i}z_{k,j}]=(1-2 \cdot \theta(X_i,X_j)/\pi)^q$.
\end{center}
Now let $\theta_r=\cos^{-1}(1-r^2/2)$, 
$$\|X_i-X_j\| \leq r \implies \theta(X_i,X_j)\leq \theta_r \implies \EE[\langle Z_i,Z_j \rangle] \geq d (1-2\theta_r/\pi)^q
$$
$$
\|X_i-X_j\| \geq (1+\epsilon)r \implies \theta(X_i,X_j)\geq (1+\epsilon/2) \theta_r \implies \EE[\langle Z_i,Z_j \rangle] \leq d (1-2 (1+\epsilon/2)\theta_r/\pi)^q.
$$

{Notice that}, 
$$
\frac{ (1-2 (1+\epsilon/2)\theta_r/\pi)^q}{ (1-2\theta_r/\pi)^q}<1-\epsilon/10,
$$
for $q=\floor*{{\pi}/{(2\theta_r)}}$ and since $n^{-0.9}\leq r\leq 0.2$. 
Notice that $d(1-2\theta_r/\pi)^{q}\in [0.3d,0.5d]$. Hence, if $\|X_i-X_j\|\leq r$ and 
$\|X_l-X_k\|\geq (1+\epsilon )r$, 
$$
\EE[\langle Z_l,Z_k \rangle]<(1-\epsilon/10)
{ \EE[\langle Z_i, Z_j \rangle]} \leq  \EE[\langle Z_i, Z_j \rangle]-0.3 d \epsilon/10,
$$

By a union bound over Chernoff bounds, since $d=\log^3n$, with probability $1-o(1/poly(n))$, the inner products between any two columns of $Z$ differs from their expectations by $o(d)$. After performing the scaling procedure, and due to the fact that $d(1-2\theta_r/\pi)^{q}\leq 0.5d$, we conclude that it suffices to compute 
{\tt Inner Product ApprxNet} with $\rho= (1-2 \cdot \theta_r/\pi)^q$ and approximation error $\epsilon/100$. 

The runtime of all components of the algorithm aside from the calls to {\tt Inner Product ApprxNet} is bounded by $\tilde{O}(n/\cos^{-1}(1-r^2/2))=\tilde{O}(n^{1.9})$.
\hfill \qed

\section{Proof of Theorem \ref{DelFar}}

We present a randomized approximation algorithm which, given a pointset in $\RR^d$ and distance parameter $r$, returns the points that have at least one neighbor at distance at most $r$.

\begin{framed}
{\tt DelFar}\\

Input: Matrix $X=[x_1,\ldots,x_n]$ with each $x_{i} \in \RR^d$, parameter $r \in \RR$, constant $\epsilon \in (0, 1/2]$.

Output: $F' \subseteq \{x_1,\ldots,x_n\}$.
\begin{itemize}
\item  Let $Y$, $r'$ be the output by algorithm {\tt Standardize} on input $X$, $r$ with parameter $\epsilon/4$.
\item If $r \geq 1/n^{0.9}$ run {\tt ApprxNet(Large radius)} on input $Y$,  $\epsilon/4, r$ and return points which correspond to the set $F' \gets X\backslash F$.
\item If $r < 1/n^{0.9}$ run {\tt ApprxNet(Small radius)} on input $Y$,  $\epsilon/4, r$ and return points which correspond to the set $F' \gets X\backslash F$.
\end{itemize}
\end{framed}

By Theorems \ref{ThmLargeRadius}, \ref{ThmSmallr}, \ref{Standardize},  both {\tt ApprxNet(Large radius)} and {\tt ApprxNet(Small radius)} return a set $F$, the subset of the centers of $r$-net that are isolated, i.e. the points that do not have any neighbor at distance $(1+\epsilon)r$. Also, both procedures run in $\tilde{O}(dn^{2-\Theta(\sqrt{\epsilon})})$. Thus, {\tt DelFar} on input a $d\times n$ matrix $X$, a radius $r \in \RR$ and a fixed constant $\epsilon \in (0,1/2]$ returns a set $F'\subseteq \{x_1,\ldots,x_n\}$, which contains all the points (vectors) of $X$ that have at least one neighbor at distance $r$. Additionally, the algorithm costs $\tilde{O}(dn^{2-\Theta(\sqrt{\epsilon})})$ time and succeeds with  probability $1-o(n^{0.04})$.

\section{A general framework for high dimensional distance problems}\label{Aframework}

In this section, we modify a framework  originally introduced by \cite{HR15}, which provides an efficient way for constructing approximation  algorithms for a variety of well known distance problems. 
We present the algorithm Net and Prune of \cite{HR15}, modified to call the algorithms {\tt ApprxNet} and {\tt DelFar}. We claim that this algorithm computes, with high probability, a constant spread interval and  costs ${O}(dn^{1.999999})$ time.

 {We assume the existence of a fast approximate decider procedure for the problems we want to address using this framework, specifically an algorithm that runs in $\tilde{O}(dn^{2-\Theta(\sqrt{\epsilon})})$, where $\epsilon$ is the approximation factor. Formally,}
\begin{dfn}
Given a function $f : X \rightarrow \RR$, we call a decider procedure a $(1+\epsilon)$-decider for $f$, if
for any $x\in X$ and $r > 0$, decider$(r, x)$ returns one of the following: (i) $f(x) \in [\alpha, (1+\epsilon) \alpha]$, where $\alpha$ is some
real number, (ii) $f(x) < r$, or (iii) $f(x) > r$.
\end{dfn}
Additionally, we assume the problems we seek to improve with this method have the following property: if the decider  returns that the optimal solution is smaller than a fixed value $r$, we can efficiently remove all points that do not have any neighbor at distance at most $r$ and this does not affect the optimal solution. Let us denote $f(X)$ the optimal solution of a problem for input $X$.

\begin{framed}
{\tt Net \& Prune}\\

Input: An instance $(X, \Gamma)$ s.t. $X \subseteq \RR^d$.

Output: An interval $[x, y]$ containing the optimal value.
\begin{itemize}
\item $X_0=X, i=0$
\item While TRUE do
\begin{itemize}
\item Choose at random a point $x \in X_i$ and compute its nearest neighbor distance, $l_i$
\item Call $\frac{3}{2}$-decider$(2l_i/3, X_i)$ and $\frac{3}{2}$-decider$(cl_i, X_i)$. Do one of the following:
\begin{itemize}
\item If $\frac{3}{2}$-decider$(2l_i/3, X_i)$ returns $f(X_i) \in [x,y]$, return $f(X) \in [x/2,2y]$
\item If $\frac{3}{2}$-decider$(cl_i, X_i)$ returns $f(X_i) \in [x',y']$, return $f(X) \in [x'/2,2y']$
\item If $2l_i/3$ is too small and $cl_i$ too large, return $[l_i/3,2cl_i]$
\item If $2l_i/3$ is too large, call $X_{i+1}=${\tt DelFar}$(2l_i/3, X_i, \frac{3}{2})$ 
\item If $cl_i$ is too small,  $X_{i+1}=${\tt ApprxNet}$(4l_i, X_i, \frac{3}{2})$ 
\end{itemize}
\item $i=i+1$, 
\end{itemize}
\end{itemize}
\end{framed}

Let us denote as $|{X_i}^{\leq l}|$ and $|{X_i}^{\geq l}|$ the set of points in $X$, whose nearest neighbor distance is smaller than $l$ and greater than $l$, respectively.

\begin{thm}
Assume that the {\tt DelFar} algorithm and the {\tt ApprxNet} algorithm succeed with probability $1-\frac{1}{n^{0.01}}$. The algorithm Net \& Prune $(X, \Gamma)$ runs in expected ${O}(dn^{1.999999})$ time.
\end{thm}

\begin{proof}
In each iteration of the while loop the algorithm calls on input $X_i$ the $\frac{3}{2}$-decider procedure and either {\tt ApprxNet} or {\tt DelFar}, all of which cost ${O}(d\abs{X_i}^{1.999999})$ time. Thus, the total running time of the algorithm is ${O}(\sum_{i=0}^{i=k-1} d{|X_i|}^{1.999999} )$, where k denotes the last iteration of the while loop.

In the $(i+1)$th iteration of the while loop, where $(i+1<k)$, lets assume that $x_1, x_2,\dots,x_m$ is the points' labels in increasing order of their nearest neighbor distance in $X_i$. If j is the index of the chosen point on the first step of the algorithm and ${X_i}^{\geq j}$ and ${X_i}^{\leq j}$ are the subsets of points with index $\geq j$ and $\leq j$, respectively, then we call $i$ a successful iteration when $j \in [m/4,3m/4]$. Then, we have that $|X_i^{\geq j}| \geq |X_{i+1}|/4$ and $|X_i^{\leq j}| \geq |X_{i+1} |/4$ for a successful iteration. The probability that $i+1$ is a successful iteration is $1/2$.

At each iteration, but the last, either {\tt ApprxNet} or {\tt DelFar} gets called. Thus, for any successful iteration, a constant fraction of the point set is removed (it follows from Lemma 3.2.3 in \cite{HR15} and Theorem \ref{DelFar}).
Also, the algorithms $(1+\epsilon)$-decider, {\tt ApprxNet} and {\tt DelFar} succeed at every call with probability $1-\frac{O(\log n)}{n^{0.01}} = 1-o(1)$, since the expected number of iterations is $O(\log n)$. Hence, the expected running time of the algorithm is ${O}(dn^{1.999999})$, given the above algorithms succeed.
\end{proof}

At every step, either far points are being removed or we net the points. If the {\tt DelFar} algorithm is called, then with small probability we remove a point which is not far. This obviously affects the optimal value, thus we will prove the correctness of the algorithm with high probability. On the other hand, if the {\tt ApprxNet} algorithm is called, the net radius is always significantly smaller than the optimal value, so the accumulated error in the end, which is proportional to the radius of the last net computation, is also much smaller than the optimal value. For the following proofs we assume both {\tt DelFar} and {\tt ApprxNet} algorithms succeed, which occurs with probability $1-o(1)$.

\begin{lem}\label{drift}
For every iteration $i$, we have $\abs{f(X_i)-f(X_0)} \leq 16l_i$.
\end{lem}
\begin{proof} Let $I$ be the set of indices of the {\tt ApprxNet} iterations up to the $i$th iteration. Similarly, let $I'$ be the set of iterations where {\tt DelFar} is called. \\
If {\tt ApprxNet} was called in the $j$th iteration, then $X_j$ is at most a $6l_j$-drift of $X_{j-1}$, therefore $\abs{f(X_j)-f(X_{j-1})} \leq 12l_j$. Also, if {\tt DelFar} is called in the $j$th iteration, then $f(X_j)=f(X_{j-1})$ (by Theorem \ref{DelFar}). Let m=maxI, we have that,
\[ \abs{f(X_i)-f(X_0)} \leq \sum^i_{j=1}\abs{f(X_j)-f(X_{j-1})}= \sum_{j\in I}\abs{f(X_j)-f(X_{j-1})}+\sum_{j \in I'}\abs{f(X_j)-f(X_{j-1})}\] \[\leq \sum_{j\in I} 12l_j + \sum_{j \in I'} 0 \leq 12l_m \sum^{\infty}_{j=0} {\left(\frac{1}{4}\right)}^j \leq 16l_m \leq 16l_i\]
,where the second inequality holds since for every $j<i$, in the beginning of the $j$th iteration of the while loop, the set of points $X_{j-1}$ is a subset of the net points of a $4l_i$-net, therefore $l_j \geq 4l_i$.
\end{proof}

\begin{lem}\label{TotDrift}
For any iteration $i$ of the while loop such that {\tt ApprxNet} gets called, we have $l_i \leq f(X_0)/ \eta$, where $\eta=c-16$.
\end{lem}
\begin{proof} We will prove this with induction. Let $m_1,m_2, \ldots,m_t$ be the indices of the iterations of the while loop in which {\tt ApprxNet} gets called.\\
Base: In order for {\tt ApprxNet} to get called we must have $\eta l_{m_1} < cl_{m_1}<f(X_{m_1-1})$ and since this is the first time {\tt ApprxNet} gets called we have $f(X_{m_1-1})=f(X_0)$. Therefore, $\eta l_{m_1} < f(X_0)$.\\
Inductive step: Suppose that $l_{m_j}  \leq f(X_0)/ \eta$, for all $m_j < m_i$. If a call to $\frac{3}{2}$-rNet is made in iteration $m_i$ then again $cl_{m_i} < f(X_{(m_i)-1})=f(X_{m_{i-1}})$. Thus, by the induction hypothesis and Lemma \ref{drift} we have,
\[ l_{m_i}< \frac{f(X_{m_{i-1}})}{c} \leq \frac{f(X_0)+16l_{m_{i-1}}}{c} \leq \frac{f(X_0)+16f(X_0)/ \eta}{c}= \frac{1+16/ \eta}{c}f(X_0)=f(X_0)/ \eta \]
\end{proof}

Therefore, if we set $c=64$ we have $\eta=48$, thus by Lemma \ref{drift} and Lemma \ref{TotDrift},  \[ \abs{f(X_i)-f(X_0)} \leq 16l_i \leq 16 f(X_0)/\eta =f(X_0)/3\]

\begin{cor}\label{finalNP}
For $c \geq 64$ and for any iteration i we have:
\begin{itemize}
\item $(2/3)f(X_0) \leq f(X_i) \leq (4/3)f(X_0)$,
\item if $f(X_i) \in [x,y]$, then $f(X_0)  \in [(3/4)x,(3/2)y] \subseteq [x/2,2y] $,
\item if $f(X_0) >0$ then $f(X_i) >0$.
\end{itemize}
\end{cor}

\begin{thm} \label{ConSpread}
For $c \geq 64$, the Net \& Prune algorithm computes in ${O}(dn^{1.999999})$ time a constant spread interval containing the optimal value $f(X)$, with  probability $1-o(1)$.
\end{thm}
\begin{proof}
Consider the iteration of the while loop at which Net \& Prune terminates. If the interval $[x,y]$ was computed by the $\frac{3}{2}$-decider, then it has spread $\leq \frac{3}{2}$. Thus, by Corollary \ref{finalNP} the returned interval $[x',y'] = [x/2,2y]$ contains the optimal value and its spread is $\leq 6$. Similarly, if $2l_i/3$ is too small and $cl_i$ too large, then the returned interval is $[\frac{l_i}{3}, 2cl_i]$ and its spread is $384$.
\end{proof}

\subsection{Proof of Theorem \ref{KND}}

\begin{proof}
For this particular problem, the optimal solution is not affected by the {\tt DelFar}'s removal of the points with no other point at distance at most $r$. Also, each time the {\tt ApprxNet} algorithm is called, for a fixed distance $r$, the drift of the optimal solution is at most $2r$. 
Thus, Theorem \ref{ConSpread} holds, and we compute a constant spread interval $[x,y]$ containing the optimal value, with high probability. We then apply binary search on values $x,(1+\epsilon)x,(1+\epsilon)^2x,\ldots,y$ using the algorithm {\tt $k$th NND Decider}. We perform $O(1/\log(1+\epsilon))=O(1/\epsilon^2)$ iterations, hence the total amount of time needed is $\tilde{O}(dn^{2-\Theta(\sqrt{\epsilon})})$ and the algorithm succeeds with high probability $1-o(1)$. 
\end{proof}

\end{document}